 \documentclass[twocolumn,twoside]{IEEEtran}
 \usepackage{epsfig,epstopdf,color,subfigure,graphicx}
 \usepackage{enumerate,stfloats,url,algpseudocode}

\usepackage{gensymb}
\usepackage{graphicx}
\usepackage{bm}

\usepackage{mathrsfs}

\usepackage{amsmath,amssymb,amsthm,wasysym,latexsym}
\usepackage[boxruled]{algorithm2e}
\makeatletter
\renewcommand{\algocf@caption@boxruled}{%
			\unhbox\algocf@capbox\hfill
}%
\renewcommand{\@algocf@capt@plain}{above}
\makeatother

\SetKwInput{KwInput}{Input}
\SetKwInput{KwOutput}{Output}

\usepackage{cite}

\DeclareMathOperator{\diag}{diag}

\newtheorem{proposition}{Proposition}
\theoremstyle{remark}\newtheorem{remark}{Remark}

\newcommand{\minimize}{{\rm minimize}}

\def\ccalH{{\ensuremath{\mathcal H}}}
\def\ccalT{{\ensuremath{\mathcal T}}}

\begin{document}
	
	\title{Power System State Estimation via Feasible Point Pursuit: Algorithms and Cram{\' e}r-Rao Bound}
	
	
	\author{
		Gang Wang, \IEEEmembership{Student Member,~IEEE,}
		Ahmed S. Zamzam, \IEEEmembership{Student Member,~IEEE,}\\
		Georgios B. Giannakis,~\IEEEmembership{Fellow,~IEEE,} and
		Nicholas D. Sidiropoulos,~\IEEEmembership{Fellow,~IEEE}
		\thanks{The work of G. Wang and G. B. Giannakis was supported in part by NSF grants 1423316, 1442686, 1508993, and 1509040. The work of A. S. Zamzam and N. D. Sidiropoulos was supported in part by NSF grants 1231504 and 1525194. G. Wang, A. S. Zamzam, G. B. Giannakis, and N. D. Sidiropoulos are with the Digital Tech. Center and the ECE Dept., U. of Minnesota, Mpls, MN 55455, USA.
			G. Wang is also with State Key Lab. of Intelligent Control and Decision of Complex Systems, Beijing Inst. of Tech., Beijing 100081, China.
			E-mails: \{gangwang,\,ahmedz,\,georgios,\,nikos\}@umn.edu. 
			}
	}

	\maketitle
\begin{abstract}
Accurately monitoring the system's operating point is central to the reliable and economic operation of an electric power grid. Power system state estimation (PSSE) aims to obtain complete voltage magnitude and angle information at each bus given a number of system variables at selected buses and lines. Power flow analysis is a special case of PSSE, and amounts to solving a set of noise-free power flow equations. Physical laws dictate quadratic relationships between available quantities and unknown voltages, rendering general instances of power flow and PSSE nonconvex and NP-hard. Past approaches are largely based on gradient-type iterative procedures or semidefinite relaxation (SDR). Due to nonconvexity, the solution obtained via gradient-type schemes depends on initialization, while SDR methods do not perform as desired in challenging scenarios. This paper puts forth novel \emph{feasible point pursuit} (FPP)-based solvers for power flow and PSSE, which iteratively seek feasible solutions for a nonconvex quadratically constrained quadratic programming (QCQP) reformulation of the weighted least-squares (WLS) problem. Relative to the prior art, the developed solvers offer superior performance at the cost of higher complexity. Furthermore, they converge to a stationary point of the WLS problem. As a baseline for comparing different estimators, the Cram{\' e}r-Rao lower bound (CRLB) is derived for the fundamental PSSE problem in this paper. Judicious numerical tests on several IEEE benchmark systems showcase markedly improved performance of our FPP-based solvers for both power flow and PSSE tasks over popular WLS-based Gauss-Newton iterations and SDR approaches.
\end{abstract}
	
	
	\section{Introduction}
	\label{sec:intro}

	Recognized as the greatest engineering achievement of the $21$st century 
	\cite{nae-report}, the electric power grid is a complex cyber-physical system comprising multiple subsystems, each with a transmission infrastructure to deliver electricity
	from power generators to distribution networks to customers.
	Accurately monitoring the operational condition of a power grid is crucial to various system control and optimization tasks, which include unit commitment, optimal power flow (OPF), and economic dispatch~\cite{AburExpositoBook,book2016conejo}. To enable such an accurate monitoring, a set of system variables are specified (and enforced) or measured at selected buses and lines for determining or estimating the system's operating point, namely complex voltages at all buses of the grid. These two tasks correspond to the so-termed power flow analysis and power system state estimation (PSSE), respectively.
	Both are central to monitoring, control, and future planning of electricity networks.
	
	
	In power engineering, power flow analysis is a numerical analysis of the normal steady-state flow of electric power over the grid, that is crucial for planning future power system expansions (e.g., designing components such as generators, lines, transformers, and capacitors), as well as in determining the best operation of the existing systems~\cite{book2016conejo}. The goal of power flow analysis is to obtain complete voltage magnitude and angle information at each bus for specified or enforced load and generator active power and voltage conditions~\cite{book2016conejo}. Once this information is available, other system variables including active and reactive power flows as well as generator reactive power outputs can be analytically obtained.

	Power flow analysis amounts to solving a set of quadratic equations given by the nonlinear AC power flow model obeying Ohm's and Kirchhoff's laws. Solving power flow equations for both transmission and distribution systems is known to be NP-hard \cite{acfeasibility}. 
	Due to the nonlinear nature,
	several numerical solvers have been developed to obtain a solution that is within an acceptable tolerance \cite{book2016conejo,psse2016madani}. 
	Past solvers include the Gauss-Seidel and Newton-Raphson iterative algorithms~\cite{book2016conejo}, and the semidefinite relaxation (SDR)~\cite{psse2016madani,psse2016zhang}.
	The Gauss-Seidel method is reported as the earliest devised power flow solver~\cite{book2016conejo}. On the other hand,
	the Newton-Raphson algorithm iteratively seeks improved approximations to the zeros of real-valued functions, featuring quadratic convergence whenever the initial point lands within a small neighborhood of the zeros~\cite{1967newton}. As convergence of both algorithms relies heavily on the initial point, they may diverge if the initialization is not reliable \cite{psse2016zhang}.
	With a carefully designed objective function and sufficiently small angle differences across lines, the SDR approaches have been shown capable of recovering the true power flow solution provided that the set of available specifications includes all voltage magnitudes, and the active power flows over a spanning tree of the network~\cite{psse2016madani,psse2016zhang}.

	
	The task of PSSE can be described as estimating the voltage magnitudes and angles at all buses across the network from a subset of
	supervisory control and data acquisition (SCADA)
	measurements including active and reactive power injections and flows (at both the sending and receiving ends), as well as squared voltage magnitudes~\cite{AburExpositoBook}.
	Since its appearance in the $1970$s~\cite{Schweppe70}, PSSE has
	become a prerequisite for 
	supervisory control, system planning, and economic dispatch~\cite{ Schweppe70}.
	Nonlinear SCADA measurements however, render the PSSE problem nonconvex and NP-hard in general \cite{jstsp2014zhu}.
	
	PSSE solvers so far are largely based on Gauss-Newton iterations and SDR heuristics.
	The ``workhorse'' Gauss-Newton method for nonconvex optimization has two limitations~\cite[Sec.\,1.5]{Be99}, i.e.,
	sensitivity to the initial guess, and lack of convergence guarantees.
	SDR-based approaches on the other hand solve first for a matrix variable that can be computationally expensive~\cite{jstsp2014zhu,icassp2014skgwgg,pesgm2014gwskgg,psse2016zhang,psse2016madani,icassp2017kekatos}. SDR's performance degrades when the data-size is relatively small, or when the data do not include all voltage magnitudes~\cite{psse2016zhang,psse2016madani}. For PSSE of large-scale networks, distributed Gauss-Newton and SDR implementations have been reported in \cite{jstsp2014zhu}, \cite{tps2016mll}, \cite{tsp2015cn}.
	Solving power flow equations and the PSSE can be shown equivalent to solving nonconvex QCQPs, which is generally NP-hard \cite{hardproblems}.
	Many heuristics have been put forward.
	A feasible point pursuit (FPP) algorithm developed in \cite{spl2015sidiropoulos} was shown to enjoy improved  performance over the SDR-based methods. The FPP heuristic has been employed for solving the OPF problem \cite{tsg2016zamzam}, where the resulting solver was empirically shown more effective for multi-phase transmission networks than popular SDR- and moment relaxation-based ones \cite{tsg2016zamzam}. 
	
	Building on our precursors~\cite{
	spl2015sidiropoulos,globalsip2016wzgs} and inspired by the inherent nonconvex challenge,
	the goal of this work is to develop power flow and PSSE solvers capable of attaining or approximating the global optimum at manageable computational complexity. Starting with the WLS formulation, the power flow and PSSE tasks are equivalently reformulated as a nonconvex QCQP, which can be readily tackled by FPP. 
	We further show that our FPP-based solvers converge to a stationary point of the WLS  problem. As a baseline for comparing different SE approaches, the Cram{\' e}r-Rao lower bound (CRLB) is derived for the basic PSSE problem under additive white Gaussian noise (AWGN). This is achieved by means of Wirtinger's calculus for functional analysis over complex domains. Finally, numerical experiments using several IEEE benchmark systems corroborate the superior performance of our proposed solvers over existing methods for both power flow and PSSE tasks.
	
	Regarding notation, matrices (vectors) are denoted by upper- (lower-) case boldface letters, and $ \overline{(\cdot)} $, $(\cdot)^\ccalT$, and $(\cdot)^\ccalH$ stand for complex conjugate, transpose, and conjugate-transpose, respectively. Calligraphic letters are reserved for sets, e.g., $\mathcal{N}$. Symbol ${\Re}\{\cdot\}$ (${\Im}\{\cdot\}$) takes the real (imaginary) part of a complex-valued object, and ${\rm diag}(\bm{x})$ is a diagonal matrix holding in order entries of $\bm{x}$ on its diagonal. 

	\section{System Modeling and Problem Statement}
	\label{sec:problem}
	
	An electric transmission network having $N$ nodes (buses) and $E$ edges (lines) can be represented by a graph $\mathcal{G}:=\{\mathcal{N},\mathcal{E}\}$, whose nodes
	$\mathcal{N}:=\{1,2,\ldots,N\}$ correspond to buses, and whose edges $\mathcal{E}:=\{(m,n)\}\subseteq\mathcal{N}\times\mathcal{N}$ correspond to transmission lines. For every bus $n\in\mathcal{N}$, let $V_n:=|V_n|e^{j\theta_n}$ be the nodal complex voltage, whose magnitude and phase are given by $|V_n|$ and $\theta_n$, respectively; likewise for the complex current injection $I_n:=|I_n|e^{j\phi_n}$. Let also $S_n:=P_n+jQ_n$ be the corresponding complex power injection, in which $P_n$ and $Q_n$ are the active and reactive power injection, respectively. For every line $(m,n)\in\mathcal{E}$, let $I_{mn}$ denote the complex current flowing from bus $m$ to $n$, and $S_{mn}^f:=P_{mn}^f+jQ_{mn}^f$ the complex power flow from bus $m$ to $n$ seen at the sending end, where $P_{mn}^f$ and $Q_{mn}^f$ are the active and reactive power flow, respectively; and likewise for the receiving-end (active and reactive) power flow $P_{mn}^t$ and $Q_{mn}^t$.

	The AC power flow model dictates that system variables $\{P_n\}$, $\{Q_n\}$, $\{P_{mn}^f\}$, $\{Q_{mn}^f\}$, $\{P_{mn}^t\}$, $\{Q_{mn}^t\}$, and $\{|V_n|^2\}$ are quadratic functions of the state vector $\bm{v}$. Clearly, this holds true for the squared voltage magnitude understood as $|V_n|^2=V_n \overline{V}_n$.
	To specify the relationship between power quantities and $\bm{v}$, introduce $\bm{Y}\in\mathbb{C}^{N\times N}$ to represent the bus admittance matrix, which is in general symmetric. 
	Ohm's law in conjunction with Kirchhoff's law
	reads as
	\begin{equation}\label{eq:ohm}
	\bm{i}=\bm{Y}\bm{v}.
	\end{equation}
	It is worth mentioning that $\bm{Y}$ is sparse, thus enabling efficient computations in large-size power networks, and its $(m,n)$-th entry is given by
	\begin{align}\label{eq:admittance}
	Y_{mn}:=\left\{\begin{array}{ll}
	-y_{mn},~&(m,n)\in\mathcal{E}\\
	{y}_{nn}^g+\sum_{k\in\mathcal{N}_n}y_{nk},~&m=n\\
	0,~&{\text{otherwise}}
	\end{array}\right.
	\end{align}
	where $y_{mn}$ denotes the admittance of line $(m,n)\in\mathcal{E}$, ${y}_{nn}^g$ the admittance to the ground at bus $n\in\mathcal{N}$, and $\mathcal{N}_n$ the set of neighboring buses directly connected to bus $n$. For $m\ne n$, let ${y}_{mn}^s$ be the shunt admittance at bus $m$ associated with line $(m,n)$. Recall from Ohm's and Kirchhoff's laws that the current flowing from bus $m$ to $n$ can be expressed as
	\begin{equation}
	\label{eq:imn}
	I_{mn}={y}_{mn}^s V_m+y_{mn}(V_m-V_n)
	\end{equation}
	whereby the reverse-direction current $I_{nm}$ can be given symmetrically. Due to ${y}^s_{mn} \ne 0$ in general, it holds $I_{mn}\neq -I_{nm}$.
	
	The AC model also asserts $P_n+jQ_n=V_n \overline{I}_n,\,\forall  n\in\mathcal{N}$. Appealing again to  \eqref{eq:ohm} leads to the next matrix-vector form
	\begin{equation}
	\label{eq:injection}
	\bm{p}+j\bm{q}={\rm{diag}}(\bm{v})\overline{\bm{i}}={\rm{diag}}(\bm{v})\overline{\bm{Y}}\overline{\bm{v}}
	\end{equation}
where both active and reactive power injections are quadratically related to $\bm{v}$.
	Likewise, the sending-end active and reactive power flow over line $(m,n)\in\mathcal{E}$ can be written as
	\begin{align}
	P_{mn}^f+jQ_{mn}^f&=V_m\overline{I}_{mn}\nonumber\\
	&=\left({y}_{mn}^s+\overline{y}_{mn}^s\right)V_m\overline{V}_m-\overline{y}_{mn}V_m\overline{V}_n\label{eq:flow}
	\end{align}
	where the second equality is obtained by substituting $I_{mn}$ in \eqref{eq:imn} into the first. Hence, $P_{mn}^f$ and $Q_{mn}^f$ can also be expressible as quadratic functions of $\bm{v}$. By symmetry, this quadratic relationship also holds for  $P_{mn}^t$ and $Q_{mn}^t$.

	To perform either power flow analysis or PSSE,
	a total of $L$ system variables are specified or measured by the system operator. The nonlinear AC networks have available the next seven types of quantities: $|V_n|^2$,  $P_n$, $Q_n$, $P_{mn}^f$, $Q_{mn}^f$,  $P_{mn}^t$, and $Q_{mn}^t$. If $\mathcal{N}_V$, $\mathcal{N}_P$, $\mathcal{N}_Q$, $\mathcal{E}_P^f$ ($\mathcal{E}_Q^f$), and $\mathcal{E}_P^t$ ($\mathcal{E}_Q^t$) denote the selected sets of buses/lines over which actual quantities of the corresponding type are available, the elaborated quadratic relationships prompt us to define the $L\times 1$ data vector
	$\bm{z}:=\big[\{|V_n|^2\}_{n\in\mathcal{N}_V},\;\{P_n\}_{n\in\mathcal{N}_P},\;\{Q_n\}_{n\in\mathcal{N}_Q},\; \{P_{mn}^f\}_{(m,\,n)\in\mathcal{E}_P^f},$ $\{Q_{mn}^f\}_{(m,n)\in\mathcal{E}_Q^f},\, \{P_{mn}^t\}_{(m,n)\in\mathcal{E}_P^t}, \, \{Q_{mn}^t\}_{(m,n)\in\mathcal{E}_Q^t}\big]^\ccalT\in\mathbb{R}^{L}$, whose entries can be succinctly given by
	\begin{equation}\label{eq:zell}
	z_\ell=\bm{v}^\ccalH\bm{H}_\ell\bm{v},\quad 1\le \ell \le L
	\end{equation}
	where $\{\bm{H}_\ell\}_{\ell=1}^L$ are some coefficient matrices to be specified.
	For this purpose, let $\left\{\bm{e}_n\in\mathbb{R}^{N}\right\}_{n=1}^N$ be the canonical basis of $\mathbb{R}^N$, and introduce also the admittance-dependent matrices
	\begin{align*}
	\bm{Y}_n&:=\bm{e}_n\bm{e}_n^\ccalT\bm{Y},\qquad \qquad\qquad\qquad\qquad\quad\;\,\forall n\in\mathcal{N},\\
	\bm{Y}_{mn}^f&:=(\overline{y}_{mn}+y_{mn})\bm{e}_m\bm{e}_m^\ccalT-y_{mn}\bm{e}_m\bm{e}_n^\ccalT,\quad\forall (m,n)\in\mathcal{E},\\
	\bm{Y}_{mn}^t&:=(\overline{y}_{nm}+y_{nm})\bm{e}_m\bm{e}_m^\ccalT-y_{nm}\bm{e}_m\bm{e}_n^\ccalT,\quad\forall (m,n)\in\mathcal{E}.
	\end{align*}
	For $|V_n|^2=V_n\overline{V}_n=\bm{v}^\ccalH \bm{e}_n\bm{e}_n^\ccalT\bm{v}$, it is clear that the corresponding $\bm{H}_n$ in \eqref{eq:zell} is
	\begin{equation}
	\label{eq:hv}
	\bm{H}_{V,n}:=\bm{e}_n\bm{e}_n^\ccalT\succeq\bm{0},\quad \forall n\in\mathcal{N}
	\end{equation}
	which are rank-$1$.
	By taking separately the real and imaginary parts of \eqref{eq:injection} and \eqref{eq:flow}, we obtain the $\{\bm{H}_\ell\}$ associated with the active and reactive power injections for all buses $n\in\mathcal{N}$
	\begin{align}\label{eq:hn}
	 &\bm{H}_{P,n}:=\frac{1}{2}\left(\bm{Y}_n+\bm{Y}_n^\ccalH\right),\quad\bm{H}_{Q,n}:=\frac{j}{2}\left(\bm{Y}_n-\bm{Y}_n^\ccalH\right)
	\end{align}
	and with sending-end and receiving-end active and reactive power flow at all lines $(m,n)\in\mathcal{E}$
	\begin{subequations}
		\label{eq:hn1}
		\begin{align}
		&\bm{H}_{P,mn}^f:=\frac{1}{2}\left(\bm{Y}_{mn}^f+\big(\bm{Y}_{mn}^f\big)^\ccalH\right)\\
		&\bm{H}_{Q,mn}^f:=\frac{j}{2}\left(\bm{Y}_{mn}^f-\big(\bm{Y}_{mn}^f\big)^\ccalH\right)\\
		&\bm{H}_{P,mn}^t:=\frac{1}{2}\left(\bm{Y}_{mn}^t+\big(\bm{Y}_{mn}^t\big)^\ccalH\right)\\
		&\bm{H}_{Q,mn}^t:=\frac{j}{2}\left(\bm{Y}_{mn}^t-\big(\bm{Y}_{mn}^t\big)^\ccalH\right).
		\end{align}
	\end{subequations}
	It is worth stressing that all $\{\bm{H}_\ell\}$ in \eqref{eq:hn} and \eqref{eq:hn1}
	are sparse, low-rank, and Hermitian, but they are non-definite in general. The power flow and PSSE problems are formulated in order next.
	
	\subsection{Power flow analysis}
	Power flow analysis deals with specified power quantities, which are enforced for optimally operating an electric power grid.
	Specifically, given $L$ perfectly known specifications $\{z_\ell\}_{\ell=1}^L$ and valid network parameters $\{\bm{H}_\ell\}_{\ell=1}^L$ as in \eqref{eq:zell}, the goal of power flow analysis is to decide the state vector $\bm{v}\in\mathbb{C}^N$ that satisfies all specifications, namely,
	\begin{subequations}\label{eq:pf}
		\begin{align}
		{\rm find}~&~~{\bm{v}}\in\mathbb{C}^N\label{eq:find}\\
		{\rm subject~to}~&~~ \bm{v}^\ccalH\bm{H}_{\ell}\bm{v}=z_\ell,\quad 1\le \ell \le L\label{eq:noiseless}.
		\end{align}
	\end{subequations}
	
	Recall that each bus in a power system is classified as a PQ, PV, or slack (reference) bus based on the constraints imposed per bus. PQ buses, which often correspond to loads, specify and enforce only active and reactive power injection $P_n$ and $Q_n$ on bus $ n $. On the other hand, the PV buses, which are typically associated with generators, enforce active power injection $P_n$ and voltage magnitude $|V_n|$. For the slack bus, its voltage phase is fixed at $\theta_n=0$, by convention. With $\theta_n=0$, the power flow problem in \eqref{eq:pf} is equivalent to solving for $2N-1$ real-valued unknowns from $L$ quadratic equations.
	The classical power flow problem
	considers the particular case where the $L=2N-1$  specifications are enforced only at the PV, PQ, and slack buses as opposed to a combination of buses and lines~\cite{psse2016madani}.

	\subsection{Power system state estimation}
	PSSE on the other hand deals with noisy observations acquired by the SCADA system adhering to 
	\begin{equation}\label{eq:noisy}
	z_\ell=\bm{v}^\ccalH\bm{H}_\ell\bm{v}+\eta_\ell,\quad 1\le\ell \le L
	\end{equation}
	where $\eta_\ell$ accounts for the zero-mean 
	distributed measurement error with known variance $\sigma_\ell^2$, henceforth assumed independent across meters.
	The goal of PSSE is, given SCADA measurements $\{z_\ell\in\mathbb{R}\}_{\ell=1}^L$ and also parameters $\{\bm{H}_\ell\}_{\ell=1}^L$, estimate the state vector $\bm{v}\in\mathbb{C}^N$.

	Adopting the WLS criterion,
	the SE task can be cast as that of solving the following nonlinear LS problem
	\begin{equation}
	\hat{\bm{v}}:=\arg\min_{\bm{v}\in\mathbb{C}^N}\sum_{\ell=1}^L w_\ell\!\left(z_\ell-\bm{v}^\ccalH\bm{H}_\ell\bm{v}\right)^2\label{eq:psse}
	\end{equation}
	where entries of the weight vector $\bm{w}:=[w_1~\cdots~w_L]^\ccalT$ are often taken as $w_\ell:=1/\sigma_\ell^2$ for known $\sigma_\ell^2$ values.
	The WLS estimate $\hat{\bm{v}}$ coincides with the maximum likelihood one when the error vector $\boldsymbol{\eta}:=[\eta_1~\cdots~\eta_L]^\ccalT$ obeys the multivariate Gaussian distribution
	$\mathcal{N}(\bm{0},{\rm diag}(\boldsymbol{\sigma}^2))$ with $\bm{\sigma}^2:=[\sigma_1^2~\cdots~\sigma_L^2]^\ccalT$.
	Unfortunately, due to the quadratic terms $\{\bm{v}^\ccalH\bm{H}_\ell\bm{v}\}$ inside the squares, the WLS SE problem is nonconvex. Minimizing nonconvex objectives, which typically exhibit many stationary points, is \emph{NP-hard} in general \cite{hardproblems}. 
	Hence, solving the problem in \eqref{eq:psse} is indeed challenging.


	PSSE approaches so far can be grouped as convex and nonconvex ones. The latter includes the ``workhorse''  Gauss-Newton method, which is also typically employed in practice: Upon linearizing the error function in the LS cost around a given estimate, the minimizer of the norm of the resulting linearized approximation is used to initialize the next iteration~\cite[Sec.\,1.5]{Be99}. Minimizing nonconvex functions, Gauss-Newton iterations can be problematic due to: i) its sensitivity to the initial point; and, ii) lack of convergence guarantee to even a stationary point~\cite{Be99}.
	Convex approaches via SDR~\cite{jstsp2014zhu,psse2016zhang}
	express all data  $\{z_\ell\}$ as linear functions of the outer-product $\bm{V}:=\bm{v}\bm{v}^\ccalH\in\mathbb{C}^{N\times N}$.
	Problem \eqref{eq:psse} is then convexified by dropping the nonconvex constraint ${\rm rank}(\bm{V})=1$.
	SDR-based methods seldom yield solutions of rank-$1$ in the noisy case.
	Further eigen-decomposition or randomization procedures are required to recover the estimator $\hat{\bm{v}}$ from the SDR solution $\hat{\bm{V}}$.
	Performance of SDR solutions degrades when the data size is small, or when the set of measurements does not include the voltage magnitude at all buses,
	as will be demonstrated by our numerical results in Sec. \ref{sec:test}.

	\section{Feasible Point Pursuit based Solvers}\label{sec:fpp}
	
	In this section, the FPP-based power flow and PSSE solvers will be developed based on procedures distinct from existing iterative optimization and SDR-based SE approaches. To this end, some basics of FPP are first reviewed.
	For nonconvex QCQPs,
	FPP iteratively solves a series of convexified QCQPs obtained with successive convex inner-restrictions of the original nonconvex feasibility set, and with additive slacks
	to approximate the feasible solutions of the original nonconvex QCQP~\cite{spl2015sidiropoulos}. Specifically, starting with an initial guess, FPP first decomposes the quadratic terms in all nonconvex constraints into their convex and nonconvex parts by means of eigen-decomposition, which can be efficiently carried out offline; then it linearizes the nonconvex parts around the current iterate to obtain a restricted convex QCQP. Due to restriction of the feasibility set, the convexified QCQP may be infeasible. To sustain feasibility,
	a slack variable is introduced for each relaxed constraint,
	with a convex penalty on the slack variables added to the cost function, which can
	enforce sparing use of slacks to produce solutions of minimal constraint violation. The minimizer of the regularized convex QCQP subproblem is taken as the next iterate, which will be used as the linearization point of the nonconvex components at the next iteration. This successive convex approximation and feasibility-restoring procedure is repeated until a certain stopping criterion is met. Further details of FPP can be found in \cite{spl2015sidiropoulos}, \cite{tsg2016zamzam}.
	
	Note that the power flow problem  \eqref{eq:pf} consists of quadratic equality constraints, which are not in the standard QCQP form. To apply FPP,  equalities are relaxed to inequalities, while penalizing the slack variables $\bm{s}:=\{s_\ell\ge 0\}_{\ell=1}^L$, yielding
	\begin{subequations}		\vspace{-1.2em}
		\label{eq:pfqcqp}
		\begin{align}
		\underset{\bm{v}\in\mathbb{C}^N,\,
			\{s_\ell\}_{l=1}^{L}}{\minimize}
		~&~~f(\bm{s})={\sum_{l=1}^{L}  s_\ell^2}
		\\
		{\rm subject~to}	~&~~	\left|z_\ell - \bm{v}^\ccalH \bm{H}_{\ell} \bm{v}\right | \leq s_\ell, \quad  1\le \ell\le L\label{eq:pfqcqp2}
		\end{align}
	\end{subequations}
	where other choices of the convex penalty function $f(\cdot)$ include the (weighted) $\ell_1$ or $\ell_\infty$ norm. 
Problem~\eqref{eq:pfqcqp} is equivalent to the original power flow formulation \eqref{eq:pf} when the latter is feasible. To see this, assume that the set of power flow equations in \eqref{eq:noiseless} admits (possibly more than one) feasible solutions. Clearly at the optimum of \eqref{eq:pfqcqp}, the objective reduces to zero, the slack variables $\{s_\ell\}_{\ell=1}^L$ take zero values, and all equalities in \eqref{eq:pfqcqp2} are achieved, thus yielding a feasible solution to the set of power flow equations in \eqref{eq:pf}.

	Similarly, our PSSE formulation in \eqref{eq:psse} minimizes a quartic polynomial of $\bm{v}$.
	To use FPP, problem \eqref{eq:psse} is reformulated as
	\begin{subequations}	
				\vspace{-1.2em}
		\label{eq:psseqcqp}
		\begin{align}
		\underset{\bm{v}\in\mathbb{C}^N,\;
			\{s_\ell\}_{l=1}^{L}}{\minimize}
		~&~~f(\bm{s})= {\sum_{l=1}^{L} w_\ell s_\ell^2}\label{eq:psseqcqp1}\\
		{\rm subject~to}	~&~~	\left|z_\ell - \bm{v}^\ccalH \bm{H}_{\ell} \bm{v}\right | \leq s_\ell, \quad  1\le \ell\le L\label{eq:psseqcqp2}
		\end{align}
	\end{subequations}
	where the slack variables $ \bm{s}:=\{s_\ell\ge 0\}_{\ell=1}^L $ in this case relate to the deviations between noisy measurements $\{z_\ell\}_{\ell=1}^L$ and the actual quantities $\{\bm{v}^\ccalH
	\bm{H}_\ell\bm{v}\}_{\ell=1}^L$. Problem \eqref{eq:psseqcqp} can be similarly shown equivalent to \eqref{eq:psse}. Other convex penalty functions $f(\cdot)$ in \eqref{eq:psseqcqp1} can also be selected. In particular, if the error vector follows
	the multivariate Laplace distribution, i.e., $\bm{\eta} \sim {\rm Laplace}(\bm{0},\,\bm{b})$ with $\bm{b}:=[b_1~\cdots~b_L]^\ccalT$ collecting all scaling parameters, minimizing the $\ell_1$-based function $f(\bm{s})=\sum_{\ell=1}^{L}w_\ell s_\ell$ with $w_\ell =1/b_\ell$
	in \eqref{eq:psseqcqp} produces the maximum likelihood estimate~\cite{AburExpositoBook,psse2016zhang}.

	Apparently, the reformulated power flow and PSSE problems are of the same form [cf. \eqref{eq:pfqcqp} and  \eqref{eq:psseqcqp}], except for a minor difference in the cost functions. Setting unit weights $w_\ell=1$ in \eqref{eq:psseqcqp} reduces problem \eqref{eq:psseqcqp} to  \eqref{eq:pfqcqp}.
	Without loss of generality, we will hereafter focus on the PSSE formulation \eqref{eq:psseqcqp}, and develop the novel FPP solver. The power flow problem can be readily handled with all weights being $w_\ell =1$.
	
	In this direction, let us first convert problem \eqref{eq:psseqcqp} into a standard QCQP. Note that
	constraints \eqref{eq:psseqcqp2} can be replaced by two sets of inequalities to arrive at
	\begin{subequations}
		\label{eq:qcqpstd}
		\begin{align}
		\underset{\bm{v}\in\mathbb{C}^N,\,
			\bm{s}\in\mathbb{R}^L}{\minimize}~&~~{\sum_{l=1}^{L} w_\ell s_\ell^2}\\
		{\rm subject~to}	~&~~\bm{v}^\ccalH \bm{H}_{\ell} \bm{v} \leq  z_\ell   + s_\ell ,\quad\quad\quad\;\;  1\le \ell\le L\label{eq:qcqpstd2}\\
		~&~~\bm{v}^\ccalH \left(-\bm{H}_{\ell}\right) \bm{v} \leq  -z_\ell   + s_\ell ,\quad 1\le \ell\le L \label{eq:qcqpstd3}.
		\end{align}
	\end{subequations}
	Problem \eqref{eq:qcqpstd} is
	nonconvex even for (semi)definite coefficient matrices $\{\bm{H}_\ell\}_{\ell=1}^L$. Next we demonstrate how to take advantage of FFP to solve the problem at hand in detail.

	As discussed in Sec. \ref{sec:problem}, there are two types of
	$\{\bm{H}_\ell\}$ matrices, one corresponding to the squared voltage magnitude, and the other to power quantities. Type-I $\{\bm{H}_\ell\}$ are positive semidefinite [cf. \eqref{eq:hv}], while Type-II are non-definite [cf. \eqref{eq:hn} and \eqref{eq:hn1}]. For ease of exposition, let us introduce the FPP constraint convexification procedure 
	using one nonconvex quadratic constraint in \eqref{eq:qcqpstd}. 
	Along the lines of FPP,
	consider the term $\bm{v}^\ccalH \bm{H}_{\ell} \bm{v}$ in \eqref{eq:qcqpstd2} for some $\bm{H}_\ell $ in \eqref{eq:hn}, which can be decomposed into its convex and nonconvex components as 
	\begin{equation}\label{eq:sources}
	\bm{v}^\ccalH \bm{H}_{\ell}^{(+)} \bm{v}+ \bm{v}^\ccalH \bm{H}_{\ell}^{(-)} \bm{v} \leq z_\ell + s_\ell 
	\end{equation}
	where $ \bm{H}_{\ell}^{(+)}$ and $\bm{H}_{\ell}^{(-)} $ represent the positive semidefinite (convex) and negative semidefinite (nonconvex) parts of $ \bm{H}_{\ell} $ in \eqref{eq:sources}, respectively. For the nonconvex source $ \bm{v}^\ccalH \bm{H}_{\ell}^{(-)} \bm{v}  $  in \eqref{eq:sources}, an inner linear restriction will be derived next.
	
	The following inequality holds for any $\bm{y}\in\mathbb{C}^N$ due to the negative semidefiniteness of $\bm{H}_\ell^{(-)}$
	\begin{equation}
	(\bm{v} - \bm{y})^\ccalH \bm{H}_{\ell}^{(-)} (\bm{v}-\bm{y}) \leq 0.
	\end{equation}
	Upon expanding the left-hand-side and rearranging terms, one arrives at
	\begin{equation*}
	\bm{v}^\ccalH \bm{H}_{\ell}^{(-)} \bm{v} \leq 2\Re\!\left \{\bm{y}^\ccalH \bm{H}_{\ell}^{(-)} \bm{v}\right\} - \bm{y}^\ccalH \bm{H}_{\ell}^{(-)} \bm{y}.
	\end{equation*}
	Key to the FPP algorithm is  replacing the nonconvexity stemming from $\bm{H}_\ell^{(-)}$ in \eqref{eq:sources} or \eqref{eq:qcqpstd2} by its inner linear approximation at some given point $\bm{y}$ to yield
	\begin{equation}
	\bm{v}^\ccalH \bm{H}_{\ell}^{(+)} \bm{v} + 2\Re\!\left\{ \bm{y}^\ccalH \bm{H}_{\ell}^{(-)} \bm{v}\right\} \leq z_\ell + \bm{y}^\ccalH \bm{H}_{\ell}^{(-)} \bm{y} + s_\ell.\label{FPP-Const1}
	\end{equation}
	The strategy in selecting the linearization point $\bm{y}$ will be discussed shortly.
	In the same fashion, the nonconvex quadratic constraints in \eqref{eq:qcqpstd3} can be replaced by
	\begin{equation}
	\bm{v}^\ccalH \big(\!-\!\bm{H}_{\ell}^{(-)}\big) \bm{v} \!- 2 \Re\!\left \{ \!\bm{y}^\ccalH \bm{H}_{\ell}^{(+)} \bm{v}\!\right\}\! \leq\! -z_\ell - \bm{y}^\ccalH \bm{H}_{\ell}^{(+)} \bm{y} + s_\ell.\label{FPP-Const2}
	\end{equation}
	
	Heed that the flexibility introduced by the slacks $\{s_\ell\}_{\ell=1}^L$ always restores the feasibility of the relaxed constraints, which contributes to improved performance of FPP over other convexification approaches~\cite{spl2015sidiropoulos}.
	In the presence of noise, the minimum values required for  $ \{s_\ell\ge 0\}_{\ell}^L$ to satisfy \eqref{FPP-Const1} and \eqref{FPP-Const2} depend on the measurement error contained in $\{z_\ell\}_{\ell=1}^L$.
	
	The FPP method replaces all nonconvex constraints in \eqref{eq:qcqpstd2} by their convex restriction \eqref{FPP-Const1}, and those in \eqref{eq:qcqpstd3} by \eqref{FPP-Const2} to derive a convexified QCQP regularized with slack variables to ensure feasibility. Minimizing some convex penalty function of the slacks $\{s_\ell\}_{\ell=1}^L$ not only minimizes the fitting error between
	$\{z_\ell\}$ and $\{\bm{v}^\ccalH\bm{H}_\ell\bm{v}\}$, but also
	enforces sparing use of slacks and promotes
	solutions of minimal constraint violation.
	
	In a nutshell, the developed FPP-based PSSE solver can be understood as follows. Starting with an initial point $\bm{v}_0$ (typically the flat voltage profile point, i.e., all-ones vector), our FPP-based solver successively tackles a sequence of convexified QCQPs with the linearization point being the current iterate $\bm{v}_k$, which is the $\bm{v}$-minimizer obtained by solving a convexified QCQP at the previous iteration. Hence, assuming available the $\bm{v}$-minimizer $\bm{v}_k$ at the $(k+1)$-st iteration,
	our FPP-based solver boils down to solving the following convexified QCQP subproblem
	\begin{subequations}
		\label{eq:fpp}
		\begin{align}
		&\left\{\bm{v}_{k+1},\bm{s}_{k+1}\right\}:=	
		\arg\min_{\bm{v},\,
			\bm{s}}
		\sum_{l=1}^{L} w_\ell s_\ell^2    \label{eq:fpp1}\\
		&{\rm subject~to}\nonumber\\
		&			\bm{v}^\ccalH \bm{H}_\ell^{(+)} \bm{v} \!+ 2 \Re\!\left\{\bm{y}^\ccalH \bm{H}_\ell^{(-)} \bm{v}\right\}  \leq z_\ell\! + \bm{y}^\ccalH \bm{H}_\ell^{(-)} \bm{y}\! + s_\ell\label{eq:fpp2}\\
		& \bm{v}^\ccalH \bm{H}_\ell^{(-)} \bm{v}\! + 2 \Re\!\left\{\bm{y}^\ccalH \bm{H}_\ell^{(+)} \bm{v}\right\}  \geq z_\ell \!+ \bm{y}^\ccalH \bm{H}_\ell^{(+)} \bm{y}\! -s_\ell\label{eq:fpp3}\\
		& \forall \ell =1,2,\ldots,L\nonumber
		\end{align}	
	\end{subequations}
	where $\bm{y}:=\bm{v}_k$ is the $\bm{v}$-minimizer of \eqref{eq:fpp} at the $k$-th iteration. The QCQP in \eqref{eq:fpp} is convex, which can be solved in polynomial time using off-the-shelf solvers~\cite{sedumi}

	The FPP-based PSSE solver is summarized in Alg. \ref{alg:fpp}. 	
The following three properties of our FPP-based solver are worth highlighting.
	
	\begin{remark}[\it Power flow analysis]
		\label{rmk:1}
		Cast as a special instance of PSSE, the power flow problem in \eqref{eq:pf} can be solved by our developed FPP-based PSSE solver with unit weights $w_\ell=1$.
	\end{remark}
	
	\begin{remark}[\it Bad data removal]
		\label{rmk:2}
		Besides the $\ell_2$-norm in \eqref{eq:fpp1}, other convex penalty functions can be used to fit different (noisy) data models. In particular, adopting the weighted $\ell_1$-norm (i.e., replacing $s_\ell^2$ with $|s_\ell|$) yields the weighted least-absolute-value estimator known for bad data cleansing
		\cite{CelikAbur92}. 
	\end{remark}
	
	\begin{remark}[\it Synchrophasors]
		\label{rmk:3}
		Synchrophasors, if available, can be easily incorporated into the developed PSSE formulation \eqref{eq:fpp}. To see this, letting $\bm{\zeta}_n=\bm{\Phi}_n\bm{v}+\bm{\epsilon}_n$ collect the noisy PMU data at bus $n$, hybrid estimation exploiting both nonlinear SCADA measurements and linear PMU ones can be achieved \cite{2011pmu} with an additional data-fitting term for the PMU data in \eqref{eq:fpp1}, namely, $\sum_{n\in\mathcal{P}}\|\bm{\zeta}_n -\bm{\Phi}_n \bm{v}\|_2^2$, where $\mathcal{P}$ denotes the subset of the PMU-instrumented buses.
	\end{remark}

		\SetKw{Init}{Initialization: }
		\DontPrintSemicolon
		\begin{algorithm}[t]
			\renewcommand{\arraystretch}{1}
			\caption{FPP-based power flow and PSSE Solvers}
			\vskip 0.05in
			\KwInput{Data $\{(z_\ell,\bm{H}_\ell)\}$; weights $\{w_\ell=1\}$ for power flow, and $\{w_\ell=1/\sigma_\ell^2\}$ for PSSE; solution accuracy $\epsilon>0$.}
			\Init{\rm set $ k = 0 $ and $ \bm{y}=[1~\cdots~1]^\ccalT$}.\;
			\Repeat{$\left\|\bm{v}_{k} - \bm{v}_{k-1} \right\|_2  \leq \epsilon $.}{
				$ \{\bm{v}_k, \bm{s}_k\} \leftarrow$ minimizer of problem \eqref{eq:fpp} \;
				$ \bm{y} \leftarrow \bm{v}_k $\;
				$ k \leftarrow k+1 $\;
			}
			\KwOut{$ \hat{\bm{v}} \leftarrow \bm{v}_k $.}
			\label{alg:fpp}
			\vskip -0.01in
					\end{algorithm}

On the theoretical side, the following result establishes convergence of our developed FPP-based solvers to a stationary point of the WLS formulation.
	
	\begin{proposition}[Global convergence of FPP-based solvers]
		Let $ \{ \bm{v}_k \}_{k=0}^{\infty} $ be any sequence generated by the FPP-based solver  in Alg.~\ref{alg:fpp}. Then, all limit points of $ \{ \bm{v}_k \}_{k=0}^{\infty} $ are stationary points of the WLS problem in \eqref{eq:psse}.
	\end{proposition}
	
	\begin{proof}
		As elaborated in Sec. \ref{sec:fpp},
		solving problem \eqref{eq:qcqpstd} is equivalent to  solving problem \eqref{eq:psse}. The nonconvex QCQP of complex-valued vector $\bm{v}\in\mathbb{C}^N$ in \eqref{eq:qcqpstd} can be equivalently posed as a QCQP of the expanded real-valued vector $ {\bm u} := [\Re(\bm{v})^T~ \Im(\bm{v})^T]^T\in\mathbb{R}^{2N}  $, where the associated quadratic matrices $\{ \overline{{\bm H}}_\ell\} $ are given as
		\begin{equation*}
		\overline{\bm{H}}_\ell := \begin{bmatrix}
		\Re(\bm{H}_\ell) & -\Im(\bm{H}_\ell)\\
		\Im(\bm{H}_\ell) & \Re(\bm{H}_\ell) 
		\end{bmatrix}\in\mathbb{R}^{2N\times 2N},\quad  1\leq \ell\leq L.
		\end{equation*}
		Accordingly, each constraint in \eqref{eq:qcqpstd}
		can be re-expressed as the difference between two convex functions. To see this, consider e.g. constraint~\eqref{eq:qcqpstd2}, which can be rewritten as
		\begin{equation}
		\big(\bm{u}^\ccalT \overline{\bm{H}}_{\ell}^{(+)} \bm{u} - s_\ell\big) - \big(\bm{u}^\ccalT (-\overline{\bm{H}}_{\ell}^{(-)}) \bm{u}\big)\leq  z_\ell
		\end{equation}
		where $ \overline{\bm{H}}_{\ell}^{(+)} $ and $ \overline{\bm{H}}_{\ell}^{(-)} $ are the positive and negative semidefinite parts of $ \overline{\bm{H}}_{\ell} $, hence rendering
		terms $\bm{u}^\ccalT \overline{\bm{H}}_{\ell}^{(+)} \bm{u} - s_\ell$ and $\bm{u}^\ccalT (-\overline{\bm{H}}_{\ell}^{(-)}) \bm{u}$ both convex.
		Alg. \ref{alg:fpp} is tantamount to an application of the convex-concave procedure~\cite{Yuille-2003,2017park} to the reformulated QCQP in the real domain.
		Hence, the sequence generated by Alg. \ref{alg:fpp} converges to a stationary point of \eqref{eq:psse} by appealing to the results in \cite[Thm. 10]{Lanckriet-2009}.
	\end{proof}

	\section{Cram{\' e}r-Rao Bound for PSSE}
	According to standard results from estimation theory \cite{kaybook}, the variance of any unbiased estimator is lower bounded by the Cram{\' e}r-Rao lower bound (CRLB). Appreciating its key role as a performance benchmark across different estimators,
this section establishes the
CRLB for the fundamental PSSE problem. The CRLB analysis of PSSE however, entails finding derivatives (gradient and Hessian) of a real-valued function with respect to multiple complex-valued variables. To address this challenge, we call for advanced complex analysis tools based on the so-termed Wirtinger derivative and Wirtinger's calculus, which are detailed in the Appendix. 	
	The following result provides a closed-form CRLB for any unbiased PSSE solver under the AWGN model in \eqref{eq:noisy}, which can be directly used to assess the performance of other PSSE solvers. 
	
	\begin{proposition}\label{prop:crlb}
		Consider estimating the unknown state vector $\bm{v}\in\mathbb{C}^N$ from noisy data $\{z_\ell\}_{\ell=1}^L$ obeying the model in \eqref{eq:noisy}, where the noise $\eta_\ell$ is assumed Gaussian distributed with mean zero and variance $\sigma_\ell^2$, and is also independent across meters. Then the covariance matrix of any unbiased estimator $\hat{\bm{v}}$ obeys
		\begin{equation}\label{eq:crlb}
		{\rm Cov}(\hat{\bm{v}})\succeq [\bm{F}^\dagger(\bm{v},\overline{\bm{v}})]_{1:N,1:N}
		\end{equation}
		where the Fisher information matrix is given by
		\begin{equation}\label{eq:fim}
		\bm{F}(\bm{v},\overline{\bm{v}})\!=\!\!\left[\!\!\begin{array}{cc}
		\sum_{\ell=1}^L \!\tfrac{1}{\sigma_\ell^2}(\bm{H}_\ell\bm{v})(\bm{H}_\ell\bm{v})^\ccalH
		\!&\!\!
		\sum_{\ell=1}^L\!\tfrac{1}{\sigma_\ell^2}(\bm{H}_\ell\bm{v})(\overline{\bm{H}}_\ell\overline{\bm{v}})^\ccalH\\
		\sum_{\ell=1}^L\!\tfrac{1}{\sigma_\ell^2}(\overline{\bm{H}}_\ell\overline{\bm{v}})(\bm{H}_\ell\bm{v})^\ccalH
		\!&\!\!
		 \sum_{\ell=1}^L\!\tfrac{1}{\sigma_\ell^2}(\overline{\bm{H}}_\ell\overline{\bm{v}})(\overline{\bm{H}}_\ell\overline{\bm{v}})^\ccalH
		\end{array}\!\!
		\right].
		\end{equation}
		Furthermore, $\bm{F}$ has at least rank-$1$ deficiency even when all possible SCADA measurements are available.
	\end{proposition}
	
	The proof of Prop. \ref{prop:crlb} is deferred to the Appendix. Even though the Fisher information matrix (FIM) in \eqref{eq:fim} 
	is rank deficient, the pseudo-inverse of FIM qualifies itself as a valid yet generally looser lower bound on the mean-square error (MSE) of any unbiased estimator \cite{tsp2001sm}. This lower bound is often attainable in practice, and is predictive of optimal estimator performance~\cite{tsp2001sm}, as will be demonstrated by our numerical tests in Sec. \ref{sec:test}. The derived CRLB in \eqref{eq:crlb} will be employed to benchmark and compare performance of different PSSE solvers next.
	
	\section{Simulated Tests}\label{sec:test}
	
	In the section, we compare the proposed FPP-based solvers in Alg. \ref{alg:fpp}
with existing alternatives including the WLS via Gauss-Newton iterations (GN-based), and the SDR-based solver (SDR-based)~\cite{jstsp2014zhu} for both power flow and PSSE tasks on several IEEE benchmark systems~\cite{PSTCA}. Throughout, all reported numerical results were obtained by averaging over $100$ independent Monte Carlo realizations.
	The three PSSE solvers from noisy measurements are compared in terms of the mean-square error $\sum_{i=1}^{100}\left\|\hat{\bm{v}}_i-\bm{v}\right\|_2^2/100$, where $\hat{\bm{v}}_i$ is the returned estimate at the $i$-th realization, and $\bm{v}$ the actual voltage profile.
	In the absence of noise, performance of the power flow solvers is assessed through the empirical success rate over $100$ trials. A success is declared for a trial if the returned power flow solution $\hat{\bm{v}}$ incurs a
	relative violation on the given set of $L$ power flow equations, given by $\sum_{\ell=1}^{L}(z_\ell-\hat{\bm{v}}^\ccalH\bm{H}_\ell\hat{\bm{v}})^2/\sum_{\ell=1}^Lz_\ell^2$ less than $10^{-3}$. 
	(The reason why $\|\bm{v}-\hat{\bm{v}}\|_2^2$ is not used is due to existence of possibly multiple solutions $\bm{v}$ satisfying the set of power flow equations.)

		\begin{table}[t]
			\centering
			\caption{Empirical success rate on IEEE test systems with $\theta=0.1\pi$.}
			\begin{tabular}{l|c|c|c|c|c|c}
				\hline
				Test case & $5$-bus & $9$-bus & $14$-bus & $24$-bus & $30$-bus & $39$-bus \\
				\hline
				FPP-based    & $100\%$       & $100\%$     & $100\%$      & $100\%$     & $100\%$      &$100\%$  \\
				\hline
				SDR-based    & $0$      & $29\%$   & $40\%$  & $2\%$   & $94\%$   & $97\%$  \\
				\hline
				GN-based     & $100\%$       & $100\%$   & $100\%$   & $100\%$   &$87\%$  & $64\%$  \\
				\hline
			\end{tabular}%
			\label{tab:succ1}\vspace{-.3em}
		\end{table}%
		
		\begin{table}[t] 
			\centering
			\vspace{-.4em}
			\caption{Empirical success rate on IEEE test systems with $\theta=0.3\pi$.}
			\begin{tabular}{l|c|c|c|c|c|c}
				\hline
				Test case & $5$-bus& $9$-bus & $14$-bus & $24$-bus & $30$-bus & $39$-bus \\
				\hline
				FPP-based    & $100\%$       & $100\%$     & $100\%$      & $100\%$     & $100\%$      &$100\%$  \\
				\hline
				SDR-based    & $4\%$      & $33\%$   & $15\%$  & $10\%$    & $0$   & $0$  \\
				\hline
				GN-based    & $100\%$        & $55\%$   & $33\%$   & $21\%$   &$0$  & $5\%$  \\
				\hline
			\end{tabular}%
			\vspace{-1em}
			\label{tab:succ2}
		\end{table}%
	
	Different system quantities and voltage profiles were generated via the MATLAB-based toolbox MATPOWER~\cite{MATPOWER}. The Gauss-Newton method was implemented using the SE function `doSE.m' in MATPOWER, which was modified to terminate either upon convergence, or, when the condition number of the approximate linearization exceeds $10^{5}$ flagging explosion of the iterates~\cite{jstsp2014zhu}. The SDR- and FPP-based solvers were realized via the optimization modeling package YALMIP~\cite{YALMIP}, as well as the interior-point solver SeDuMi~\cite{sedumi}.  Furthermore, the flat-voltage profile point was used as the initial guess for the Gauss-Newton and FPP approaches.
	In order to fix the phase ambiguity, the phase generated at the reference bus is set to $0$ in all tests. The FPP solver stops either when a maximum number $100$ of iterations are reached, or when the  objective value  improvement between two consecutive iterations becomes smaller $10^{-5}$.
	All experiments were conducted on an Intel CPU @ $3.4$ GHz ($32$ GB RAM) computer.
	
	To evaluate the performance of the FPP-based solver for power flow analysis, the first experiment simulates noiseless data corresponding to the classical power flow problem. That is, a total of $L=2N-1$ system variables were specified at the PV, PQ, and slack buses to solve for $2N-1$ real-valued unknowns in $\bm{v}\in\mathbb{C}^N$ with the reference bus's phase fixed at $0$. The actual voltage magnitude of each bus was uniformly distributed over $[0.9,\,1.1]$, and its angle over $[-\theta,\,\theta]$ with $\theta=0.1\pi$ and $0.3\pi$.
	Empirical success rate results on several IEEE benchmark systems were reported in Tables~\ref{tab:succ1} and \ref{tab:succ2} for $\theta=0.1\pi$ and $0.3\pi$, respectively. Apparently, our developed FPP-based power flow solver solves exactly the classical power flow problem in all simulated tests, while the SDR-based one fails with high probability. The Gauss-Newton method performs well when the initial point lies close to the actual solution due to small $\theta$ in Table \ref{tab:succ1}, while it diverges frequently for large $\theta$ values in Table \ref{tab:succ2}.

	%

	\begin{figure}[t]
				\vskip -0.1in
		\centering
		\includegraphics[scale = .6]{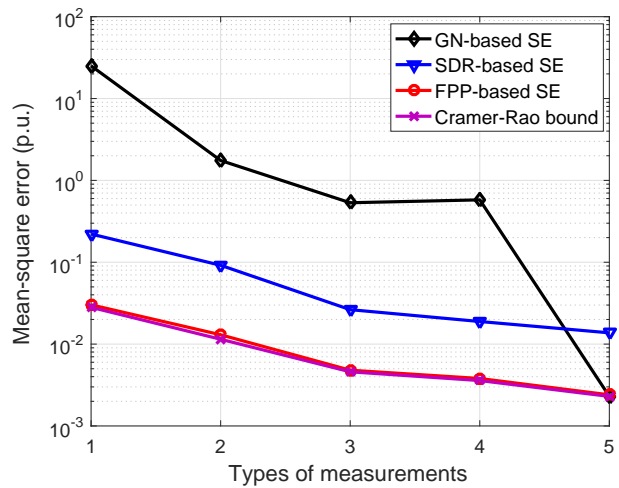} 
		\caption{MSEs as well as CRLB versus types of measurements used on the IEEE $14$-bus test system using: i) Gauss-Newton based SE; ii) SDR-based SE; and iii) FPP-based SE.
		}
		\label{fig:14bus}
		\vskip -.05in
	\end{figure}

	The second experiment compares the MSE performance of various approaches relative to the analytical Cram{\' e}r-Rao bound in \eqref{eq:crlb} on the IEEE $14$-bus test system~\cite{PSTCA}.
	The actual voltage magnitude and angle of each bus were generated uniformly over $[0.9,\,1.1]$, and $[-0.4\pi,\,0.4\pi]$, respectively.
	Initially, all voltage magnitudes as well as all sending-end and receiving-end active power flow were taken, which corresponds to the base case $3$ in the $x$-axis of Fig. \ref{fig:14bus}.
		To demonstrate the SE performance evolution of various approaches with respect to the increasing number of measurements, additional types of measurements were included in a deterministic manner detailed next.
		All seven types of SCADA measurements in \eqref{eq:hv}-\eqref{eq:hn1} were ordered as
		 $\{|V_k|^2,P_{mn}^f,P_{mn}^t,Q_{mn}^f,Q_{mn}^t,P_n,Q_n\} $.
		 Each $x$-axis value in Fig. \ref{fig:14bus} implies that the number of ordered types of measurements was used in the experiment to obtain the mean-square errors. For instance, $5$ on the $x$-axis corresponds to the case where the first $5$ types of measurements (i.e., all $|V_k|^2,P_{mn}^f,P_{mn}^t,Q_{mn}^f,Q_{mn}^t$) were used; and likewise for all other $x$-axis values.
	Measurement noise was randomly and independently generated from Gaussian distribution having zero-mean and standard deviation $ 0.1 $.
	 The SDR estimator was recovered from the SDR solution by picking the minimum-cost vector over the eigenvector
	and $5,000$ zero-mean Gaussian randomizations with covariance matrix being the SDR solution.
	The MSE as well as the CRLB versus the types of measurements available are shown in Fig.~\ref{fig:14bus}, corroborating the near-optimal performance relative to the CRLB and robustness of our developed FPP-based PSSE solver. 

	The last experiment on the IEEE $30$-bus benchmark system simulates a high signal-to-noise ratio and complete-data scenario, where all voltage magnitude as well as all active power flow at both sending- and receiving-ends were measured to be advantageous to the SDR-based method \cite{psse2016zhang}. Independent zero-mean Gaussian noise was assumed to have standard deviations $0.05$ for power measurements and $0.02$ for voltage measurements. The actual voltage magnitude and angle of each bus were generated uniformly at random over $[0.9,\,1.1]$, and $[-0.4\pi,\,0.4\pi]$, respectively.
	 Fig.~\ref{fig:30bus} depicts the average magnitude and angle estimation errors of three PSSE schemes across buses.
	The curves in Fig.~\ref{fig:30bus} demonstrate the merits of the FPP-based PSSE solver in this scenario. 

	\begin{figure}[t]
		\vskip -0.1in
		\centering
		\includegraphics[scale = 0.6]{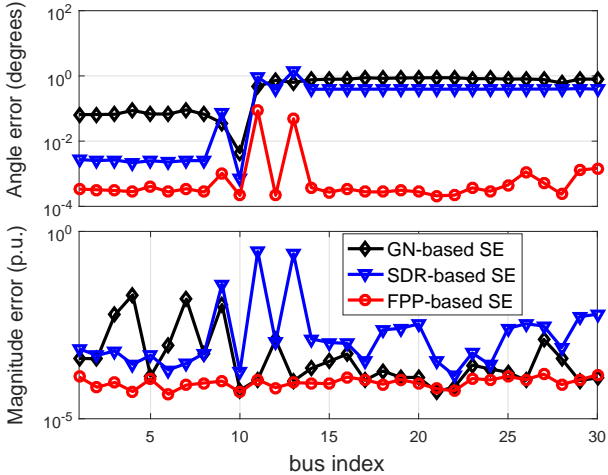} 
		\caption{Magnitude and angle estimation errors at each bus on the IEEE $30$-bus benchmark system using: i) Gauss-Newton based SE; ii) SDR-based SE; and iii) FPP-based SE.
		}
		\label{fig:30bus}
		\vskip -0.1in
	\end{figure}

	\section{Conclusions}
	
	Motivated by the inherent nonconvexity of the power flow and PSSE tasks and leveraging recent advances in handling nonconvex QCQPs,
	this work first reformulated power flow and PSSE as a nonconvex QCQP. The resulting nonconvex QCQP  was subsequently solved by the FPP algorithm. The novel FPP-based solvers were shown to converge to a stationary point of the WLS formulation. To fairly compare different PSSE solvers from noisy data, the CRLB for PSSE assuming an AWGN model was derived based on Wirtinger's calculus for functions over complex domains.
	Extensive numerical tests showed markedly improved performance of our FPP-based solver for both power flow and PSSE tasks at the price of increased runtime over competing Gauss-Newton- and SDR-based alternatives on a variety of IEEE test systems.
	
Pertinent future research directions include developing distributed implementations for large-scale power networks by exploiting the natural low-rank and sparsity structure present in the coefficient matrices $\{\bm{H}_\ell\}$. Another possibility consists of leveraging state-of-the-art approaches for tackling random quadratic systems of equations to solve the power flow and PSSE problems~\cite{taf}. 
	Generalizing feasible point pursuit algorithms to other nonconvex power grid control tasks such as stochastic energy management \cite{tps2015wang}, and distribution system-level power flow and PSSE~\cite{icassp2017kekatos} constitute meaningful directions for future research as well.

	\section{Appendix}
	
		\emph{Proof of Proposition \ref{prop:crlb}:} For the AWGN model in \eqref{eq:noisy} with $\bm{\eta}\sim \mathcal{N}(\bm{0},\diag(\bm{\sigma^2}))$, 	the data likelihood can be written as
	\begin{equation*}
	p(\bm{z};\bm{v})=\prod_{\ell=1}^L\frac{1}{\sqrt{2\pi \sigma_\ell^2}}
	\exp\bigg[-\frac{\left(z_\ell-\bm{v}^\ccalH\bm{H}_\ell\bm{v}\right)^2}{2\sigma_\ell^2}\bigg]
	\end{equation*}
	and the negative log-likelihood  $f(\bm{v})=-\ln p(\bm{z};\bm{v})$ is
	\begin{equation}\label{eq:negative}
	 f(\bm{v})=\sum_{\ell=1}^L\left[\frac{1}{2\sigma_\ell^2}\left(z_\ell-\bm{v}^\ccalH\bm{H}_\ell\bm{v}\right)^2+\frac{1}{2}\ln\left(2\pi\sigma_\ell^2\right)\right].
	\end{equation}
	
	The Fisher information matrix is defined as the Hessian of the objective function $f(\bm{v})\in\mathbb{R}$ with respect to the variable vector $\bm{v}\in\mathbb{C}^{N}$. So the task of deriving the Cram{\' e}r-Rao bound amounts to finding the Hessian of a real-valued function with respect to a complex-valued vector. Recall from \emph{Wirtinger's calculus} that $f(\bm{v})$ can be equivalently rewritten as $f(\bm{v},\overline{\bm{v}})$~\cite{wirtinger}. Upon introducing the conjugate coordinates $[\bm{v}^\ccalT~\overline{\bm{v}}^\ccalT]^\ccalT\in\mathbb{C}^{2N}$, the so-called \emph{Wirtinger derivative} is  \cite{wirtinger}
	\begin{subequations}
		\begin{align*}
		\frac{\partial f}{\partial \bm{v}}&:=\left.\frac{\partial f(\bm{v},\overline{\bm{v}})}{\partial\bm{v}^\ccalT}\right|_{\overline{\bm{v}}={\rm constant}}
		=\left.\left[\frac{\partial f}{\partial v_1}~\cdots~\frac{\partial f}{\partial v_N}\right]\right|_{\overline{\bm{v}}={\rm constant}}
		\\
		\frac{\partial f}{\partial \overline{\bm{v}}}&:=\left.\frac{\partial f(\bm{v},\overline{\bm{v}})}{\partial\overline{\bm{v}}^\ccalT}\right|_{\bm{v}={\rm constant}}	 =\left.\left[\frac{\partial f}{\partial \overline{v}_1}~\cdots~\frac{\partial f}{\partial \overline{v}_N}\right]\right|_{{\bm{v}}={\rm constant}}.
		\end{align*}	
	\end{subequations}
	Our definitions here follow the convention in multivariate calculus that derivatives are denoted by row vectors, and gradients by column vectors.
	For brevity, let $
	\phi_\ell (\bm{v}, \overline{\bm{v}}): = z_\ell - \overline{\bm{v}}^\ccalT\bm{H}_\ell\bm{v}$.
	Accordingly, the derivatives of $f$ in \eqref{eq:negative} can be obtained as
	\begin{subequations}\label{eq:derivatives}
		\begin{align}
		\frac{\partial f}{\partial \bm{v}} & = \sum_{\ell = 1}^{L} \frac{1}{\sigma_\ell^2} \phi_\ell (\bm{v}, \overline{\bm{v}}) \frac{\partial \phi_\ell (\bm{v}, \overline{\bm{v}})}{\partial \bm{v}^\ccalT}\\
		\frac{\partial f}{\partial \overline{\bm{v}}} & = \sum_{\ell = 1}^{L} \frac{1}{\sigma_\ell^2} \phi_\ell (\bm{v}, \overline{\bm{v}})\frac{\partial \phi_\ell (\bm{v}, \overline{\bm{v}})}{\partial \overline{\bm{v}}^\ccalT}
		\end{align}
	\end{subequations}
	where the partial derivatives of $ \phi_\ell $ can be found as
	\begin{subequations}\label{eq:inner}
		\begin{align}
		\frac{\partial\phi_\ell(\bm{v}, \overline{\bm{v}})}{\partial \bm{v}^\ccalT} &= -\overline{\bm{v}}^\ccalT\bm{H}_\ell=-(\bm{H}_\ell \bm{{v}})^\ccalH\label{eq:pv}\\  \frac{\partial\phi_\ell(\bm{v}, \overline{\bm{v}})}{\partial \overline{\bm{v}}^\ccalT} &=-\bm{v}^\ccalT\bm{H}_\ell^\ccalT= -(\overline{\bm{H}}_\ell \overline{\bm{v}})^\ccalH.\label{eq:pvb}
		\end{align}
	\end{subequations}

	In the conjugate coordinate system, the complex Hessian is defined as
	\begin{equation}\label{eq:hessian}
	{\boldsymbol{\mathcal{H}}}:=\nabla^2 f=\left[\begin{array}
	{cc}
	{\boldsymbol{\mathcal{H}}}_{\bm{v}\bm{v}}&{\boldsymbol{\mathcal{H}}}_{\overline{\bm{v}}\bm{v}}\\
	 {\boldsymbol{\mathcal{H}}}_{\bm{v}\overline{\bm{v}}}&{\boldsymbol{\mathcal{H}}}_{\overline{\bm{v}}\overline{\bm{v}}}
	\end{array}\right]
	\end{equation}
	whose blocks are given by
	\begin{align*}
	&	{\boldsymbol{\mathcal{H}}}_{\bm{v}\bm{v}}:=\frac{\partial}{\partial \bm{v}^\ccalT}\left(\frac{\partial f}{\partial \bm{v}}
	\right)^\ccalH, \quad {\boldsymbol{\mathcal{H}}}_{\overline{\bm{v}}\bm{v}}:=\frac{\partial}{\partial \overline{\bm{v}}^\ccalT}
	\left(\frac{\partial f}{\partial \bm{v}}
	\right)^\ccalH\\
	&	{\boldsymbol{\mathcal{H}}}_{\bm{v}\overline{\bm{v}}}:=\frac{\partial}{\partial \bm{v}^\ccalT}\left(\frac{\partial f}{\partial \overline{\bm{v}}}
	\right)^\ccalH,\quad
	{\boldsymbol{\mathcal{H}}}_{\overline{\bm{v}}\overline{\bm{v}}}:=\frac{\partial}{\partial \overline{\bm{v}}^\ccalT}\left(\frac{\partial f}{\partial \overline{\bm{v}}} \right)^\ccalH.
	\end{align*}
	After substituting \eqref{eq:derivatives} and \eqref{eq:inner} into the last equations, and
	with some tedious algebraic manipulations,
	the first block of $ {\boldsymbol{\mathcal{H}}} $ can be obtained as
	\begin{align}
	{\boldsymbol{\mathcal{H}}}_{\bm{v}\bm{v}} =& \frac{\partial}{\partial \bm{v}^\ccalT} \Big(\sum_{\ell = 1}^{L} \frac{-1}{\sigma_\ell^2} \phi_\ell (\bm{v}, \overline{\bm{v}}) \bm{H}_\ell \bm{{v}}\Big)\nonumber\\
	=& \sum_{\ell = 1}^{L} \frac{1}{\sigma_\ell^2} \Big(\bm{H}_\ell \bm{v} (\bm{H}_\ell \bm{v})^\ccalH- \phi_\ell (\bm{v}, \overline{\bm{v}}) \bm{H}_\ell\Big)\label{eq:blk11}.
	\end{align}
	The other blocks can be derived in a similar fashion. Upon omitting algebraic details, the remaining three blocks can be obtained as follows
	\begin{align}
	{\boldsymbol{\mathcal{H}}}_{\overline{\bm{v}}\bm{v}}
	&=\sum_{\ell = 1}^{L} \frac{1}{\sigma_\ell^2}  \bm{H}_\ell \bm{v} (\overline{\bm{H}}_\ell \overline{\bm{v}})^\ccalH\label{eq:blk12}\\
	{\boldsymbol{\mathcal{H}}}_{\bm{v}\overline{\bm{v}}}
	&= \sum_{\ell = 1}^{L} \frac{1}{\sigma_\ell^2}  \overline{\bm{H}}_\ell \overline{\bm{v}} (\bm{H}_\ell \bm{v})^\ccalH\label{eq:blk21}\\
	{\boldsymbol{\mathcal{H}}}_{\overline{\bm{v}}\overline{\bm{v}}}& =\sum_{\ell = 1}^{L} \frac{1}{\sigma_\ell^2} \Big(\overline{\bm{H}}_\ell \overline{\bm{v}} (\overline{\bm{H}}_\ell \overline{\bm{v}})^\ccalH- \phi_\ell (\bm{v}, \overline{\bm{v}}) \overline{\bm{H}}_\ell\Big)\label{eq:blk22}.
	\end{align}

	Evaluating the Hessian $ {\boldsymbol{\mathcal{H}}}$ in \eqref{eq:hessian} [and its blocks in \eqref{eq:blk11}-\eqref{eq:blk22}] at the true value of $\bm{v}$, and taking the expectation with respect to the noise vector $\bm{\eta}$, it is easy to verify that $ \mathbb{E} \left[ \phi_\ell (\bm{v}, \overline{\bm{v}}) \right]=0 $.
	Hence, the $\phi_\ell$-related terms disappear, so the FIM $\bm{F} := \mathbb{E} [ {\boldsymbol{\mathcal{H}}} ]\in\mathbb{C}^{2N\times 2N}$ can be expressed as \cite{van1994cramer}
	\begin{align}
	\bm{F} &=\left[
	\begin{array}{ll}
	\sum_{\ell = 1}^{L} \bm{H}_\ell \bm{v} (\bm{H}_\ell \bm{v})^\ccalH\big/{\sigma_\ell^2}&\sum_{\ell = 1}^{L}  \bm{H}_\ell \bm{v} (\overline{\bm{H}}_\ell \overline{\bm{v}})^\ccalH\big/{\sigma_\ell^2} \\
	\sum_{\ell = 1}^{L} \overline{ \bm{H}}_\ell\overline{ \bm{v}} (\bm{H}_\ell \bm{v})^\ccalH\big/{\sigma_\ell^2} &\sum_{\ell = 1}^{L}  \overline{\bm{H}}_\ell \overline{\bm{v}} (\overline{\bm{H}}_\ell \overline{\bm{v}})^\ccalH\big/{\sigma_\ell^2}
	\end{array}	\right]\nonumber\\
	&=\sum_{\ell=1}^{L}\bm{g}_\ell\bm{g}_\ell^\ccalH\buildrel\triangle\over =
	\bm{G}\bm{G}^\ccalH
	\end{align}
	where $\bm{G}:=[\bm{g}_1~\cdots~\bm{g}_L]\in\mathbb{C}^{2N\times L}$ is introduced to show the rank-deficiency of $\bm{F}$, whose $\ell$-th column is given as
	\begin{equation}
	\bm{g}_\ell:=\left[\begin{array}{c}
	\bm{H}_\ell\bm{v}/\sigma_\ell\\
	\overline{\bm{H}}_\ell\overline{\bm{v}}/\sigma_\ell
	\end{array}
	\right]=
	\left[\begin{array}{cc}
	\bm{H}_\ell/\sigma_\ell&\bm{0}\\
	\bm{0}&\overline{\bm{H}}_\ell/\sigma_\ell
	\end{array}
	\right]\left[\begin{array}{c}
	\bm{v}\\
	\overline{\bm{v}}
	\end{array}
	\right].
	\end{equation}
	
	To demonstrate the rank-$1$ deficiency of $\bm{F}$, it suffices to find a nonzero vector $\bm{d}\in\mathbb{C}^{2N}$ such that $\bm{F}\bm{d}=\bm{0}$. To this end, consider the vector $\bm{d}:=\left[\bm{v}^\ccalT~-\overline{\bm{v}}^\ccalT\right]^\ccalT\ne \bm{0}$. It is straightforward to check that for all $\ell=1,\,2,\,\ldots,\,L$
	\begin{align*}
	\bm{g}_\ell^\ccalH\bm{d}&=\Big[
	\bm{v}^\ccalH\bm{H}_\ell/\sigma_\ell ~~\overline{\bm{v}}^\ccalH
	\overline{\bm{H}}_\ell/\sigma_\ell
	\Big]\left[\begin{array}{c}
	\bm{v}\\
	-\overline{\bm{v}}
	\end{array}
	\right]
	=0
	\end{align*}
	therefore giving rise to $\bm{F}\bm{d}=\sum_{\ell=1}^{L}\bm{g}_\ell\left(\bm{g}_\ell^\ccalH\bm{d}\right)=\bm{0}$.
	That is, for any nonzero $\bm{v}$, there always exists a nonzero vector $\bm{d}=\left[\bm{v}^\ccalT~-\overline{\bm{v}}^\ccalT\right]^\ccalT$ lying in the null space of $\bm{F}$, hence verifying the rank-$1$ deficiency of $\bm{F}$. This concludes the proof.

	\IEEEtriggeratref{30}
	\bibliographystyle{IEEEtran}
	\bibliography{power}

\begin{thebibliography}{10}
\providecommand{\url}[1]{#1}
\csname url@samestyle\endcsname
\providecommand{\newblock}{\relax}
\providecommand{\bibinfo}[2]{#2}
\providecommand{\BIBentrySTDinterwordspacing}{\spaceskip=0pt\relax}
\providecommand{\BIBentryALTinterwordstretchfactor}{4}
\providecommand{\BIBentryALTinterwordspacing}{\spaceskip=\fontdimen2\font plus
\BIBentryALTinterwordstretchfactor\fontdimen3\font minus
  \fontdimen4\font\relax}
\providecommand{\BIBforeignlanguage}[2]{{%
\expandafter\ifx\csname l@#1\endcsname\relax
\typeout{** WARNING: IEEEtran.bst: No hyphenation pattern has been}%
\typeout{** loaded for the language `#1'. Using the pattern for}%
\typeout{** the default language instead.}%
\else
\language=\csname l@#1\endcsname
\fi
#2}}
\providecommand{\BIBdecl}{\relax}
\BIBdecl

\bibitem{nae-report}
\BIBentryALTinterwordspacing
W.~A. Wulf, ``Great achievements and grand challenges,'' \emph{The Bridge},
  vol.~30, no. 3/4, pp. 5--10, Fall 2010. [Online]. Available:
  \url{http://www.greatachievements.org/.}
\BIBentrySTDinterwordspacing

\bibitem{AburExpositoBook}
A.~Abur and A.~G\'{o}mez-Exp\'{o}sito, \emph{Power {S}ystem {S}tate
  {E}stimation: {T}heory and {I}mplementation}.\hskip 1em plus 0.5em minus
  0.4em\relax New York, NY: Marcel Dekker, 2004.

\bibitem{book2016conejo}
A.~G{\'o}mez-Exp{\'o}sito, A.~J. Conejo, and C.~Ca{\~n}izares, \emph{Electric
  {E}nergy {S}ystems: {A}nalysis and {O}peration}.\hskip 1em plus 0.5em minus
  0.4em\relax CRC Press, 2016.

\bibitem{acfeasibility}
K.~Lehmann, A.~Grastien, and P.~Van~Hentenryck, ``{AC}-feasibility on tree
  networks is {NP}-hard,'' \emph{IEEE Trans. Power Syst.}, vol.~31, no.~1, pp.
  798--801, Jan. 2016.

\bibitem{psse2016madani}
R.~Madani, M.~Ashraphijuo, J.~Lavaei, and R.~Baldick, ``Power system state
  estimation with a limited number of measurements,'' Preprint, 2016.

\bibitem{psse2016zhang}
Y.~Zhang, R.~Madani, and J.~Lavaei, ``Conic relaxations for power system state
  estimation with line measurements,'' \emph{IEEE Trans. Control Netw. Syst.},
  2017 (to appear).

\bibitem{1967newton}
W.~F. Tinney and C.~E. Hart, ``Power flow solution by {N}ewton's method,''
  \emph{{IEEE} Trans. Power App. Syst.}, no.~11, pp. 1449--1460, Nov. 1967.

\bibitem{Schweppe70}
F.~C. Schweppe, J.~Wildes, and D.~Rom, ``Power system static state estimation:
  Parts {I, II, and III},'' vol.~89, pp. 120--135, Jan. 1970.

\bibitem{jstsp2014zhu}
H.~Zhu and G.~B. Giannakis, ``Power system nonlinear state estimation using
  distributed semidefinite programming,'' vol.~8, no.~6, pp. 1039--1050, Dec.
  2014.

\bibitem{Be99}
D.~P. Bertsekas, \emph{Nonlinear Programming}, 2nd~ed.\hskip 1em plus 0.5em
  minus 0.4em\relax Belmont, MA: Athena Scientific, 1999.

\bibitem{icassp2014skgwgg}
S.-J. Kim, G.~Wang, and G.~B. Giannakis, ``Online semidefinite programming for
  power system state estimation,'' in \emph{Proc. {IEEE} Conf. on Acoustics,
  Speech and Signal Process.}, Florence, Italy, May 2014, pp. 6024--6027.

\bibitem{pesgm2014gwskgg}
G.~Wang, S.-J. Kim, and G.~B. Giannakis, ``Moving-horizon dynamic power system
  state estimation using semidefinite relaxation,'' in \emph{Proc. {IEEE} {PES}
  {G}eneral {M}eeting}, Washington, DC, July 2014, pp. 1--5.

\bibitem{icassp2017kekatos}
S.~Bhela, V.~Kekatos, and S.~Veeramachaneni, ``Enhancing observability in
  distribution grids using smart meter data,'' \emph{IEEE Trans. Smart Grid},
  2017 (to appear).

\bibitem{tps2016mll}
A.~Minot, Y.~M. Lu, and N.~Li, ``A distributed {G}auss-{N}ewton method for
  power system state estimation,'' \emph{IEEE Trans. Power Syst.}, vol.~31,
  no.~5, pp. 3804--3815, Sept. 2016.

\bibitem{tsp2015cn}
P.~Chavali and A.~Nehorai, ``Distributed power system state estimation using
  factor graphs,'' \emph{IEEE Trans. Signal Process.}, vol.~63, no.~11, pp.
  2864--2876, Jun. 2015.

\bibitem{hardproblems}
K.~G. Murty and S.~N. Kabadi, ``Some {NP}-complete problems in quadratic and
  nonlinear programming,'' \emph{Math. Program.}, vol.~39, no.~2, pp. 117--129,
  Jun. 1987.

\bibitem{spl2015sidiropoulos}
O.~Mehanna, K.~Huang, B.~Gopalakrishnan, A.~Konar, and N.~D. Sidiropoulos,
  ``Feasible point pursuit and successive approximation of non-convex
  {QCQP}s,'' vol.~22, no.~7, pp. 804--808, Nov. 2015.

\bibitem{tsg2016zamzam}
A.~S. Zamzam, N.~D. Sidiropoulos, and E.~Dall'Anese, ``Beyond relaxation and
  {N}ewton-{R}aphson: {S}olving {AC OPF} for multi-phase systems with
  renewables,'' \emph{IEEE Trans. Smart Grid}, 2017 (to appear).

\bibitem{globalsip2016wzgs}
G.~Wang, A.~S. Zamzam, G.~B. Giannakis, and N.~D. Sidiropoulos, ``Power system
  state estimation via feasible point pursuit,'' in \emph{IEEE Global Conf.
  Signal and Inf. Process.}, Washington, D.C., USA, 2016.

\bibitem{sedumi}
J.~F. Sturm, ``Using {S}e{D}u{M}i 1.02, a {MATLAB} toolbox for optimization
  over symmetric cones,'' \emph{Optim. Method Softw.}, vol.~11, no. 1-4, pp.
  625--653, Jan. 1999.

\bibitem{CelikAbur92}
M.~K. Celik and A.~Abur, ``A robust {WLAV} state estimator using
  transformations,'' vol.~7, no.~1, pp. 106--113, Feb. 1992.

\bibitem{2011pmu}
A.~Gomez-Exposito, A.~Abur, P.~Rousseaux, A.~de~la Villa~Jaen, and
  C.~Gomez-Quiles, ``On the use of {PMU}s in power system state estimation,''
  in \emph{Proc. Power Syst. Computat. Conf.}, vol.~22, 2011.

\bibitem{Yuille-2003}
A.~L. Yuille and A.~Rangarajan, ``The concave-convex procedure,'' \emph{Neural
  Comput.}, vol.~15, no.~4, pp. 915--936, Apr. 2003.

\bibitem{2017park}
J.~Park and S.~Boyd, ``General heuristics for nonconvex quadratically
  constrained quadratic programming,'' \emph{arXiv:1703.07870}, 2017.

\bibitem{Lanckriet-2009}
G.~R. Lanckriet and B.~K. Sriperumbudur, ``On the convergence of the
  concave-convex procedure,'' in \emph{Adv. Neural Inform. Process. Syst.},
  Vancouver, B.C., Canada, Dec. 2009, pp. 1759--1767.

\bibitem{kaybook}
S.~M. Kay, \emph{Fundamentals of Statistical Signal Processing, {V}ol. I:
  {E}stimation Theory}.\hskip 1em plus 0.5em minus 0.4em\relax Prentice Hall,
  1993.

\bibitem{tsp2001sm}
P.~Stoica and T.~L. Marzetta, ``Parameter estimation problems with singular
  information matrices,'' \emph{IEEE Trans. Signal Process.}, vol.~49, no.~1,
  pp. 87--90, Jan. 2001.

\bibitem{PSTCA}
\BIBentryALTinterwordspacing
Power systems test case archive. Univ. of Washington. [Online]. Available:
  \url{http://www.ee.washington.edu/research/pstca.}
\BIBentrySTDinterwordspacing

\bibitem{MATPOWER}
R.~D. Zimmerman, C.~E. Murillo-Sanchez, and R.~J. Thomas, ``{MATPOWER}:
  {S}teady-state operations, planning and analysis tools for power systems
  research and education,'' vol.~26, no.~1, pp. 12--19, Feb. 2011.

\bibitem{YALMIP}
\BIBentryALTinterwordspacing
J.~Lofberg, ``A toolbox for modeling and optimization in {MATLAB},'' in
  \emph{Proc. of the {CACSD} Conf.}, 2004. [Online]. Available:
  \url{http://users.isy.liu.se/johanl/yalmip/.}
\BIBentrySTDinterwordspacing

\bibitem{taf}
G.~Wang, G.~B. Giannakis, and Y.~C. Eldar, ``Solving systems of random
  quadratic equations via truncated amplitude flow,'' \emph{arXiv:1605.08285},
  2016.

\bibitem{tps2015wang}
G.~Wang, V.~Kekatos, A.~J. Conejo, and G.~B. Giannakis, ``Ergodic energy
  management leveraging resource variability in distribution grids,''
  \emph{IEEE Trans. Power Syst.}, vol.~31, no.~6, pp. 4765--4775, Nov. 2016.

\bibitem{wirtinger}
K.~Kreutz-Delgado, ``The complex gradient operator and the {CR}-calculus,''
  \emph{arXiv:0906.4835}, 2009.

\bibitem{van1994cramer}
A.~Van~den Bos, ``A {C}ram{\'e}r-{R}ao lower bound for complex parameters,''
  \emph{IEEE Trans. Signal Process.}, vol.~40, no.~10, Oct. 1994.

\end{thebibliography}

\end{document}